\documentclass[runningheads]{llncs}
\usepackage{graphicx}
\usepackage{yhmath}
\usepackage{mathtools}
\DeclareMathOperator*{\Alg}{\mathcal{A}}
\usepackage[Algorithmus]{algorithm}
\usepackage[noend]{algorithmic}

\newcommand{\algrule}[1][.2pt]{\par\vskip.5\baselineskip\hrule height #1\par\vskip.5\baselineskip}

\algsetup{
    linenodelimiter={.}
  }
\spnewtheorem{observation}{Observation}{\bfseries}{\itshape}
\usepackage{csquotes}
\usepackage{color}
\usepackage{hyperref}

\begin{document}
\title{Evacuating Two Robots from a Disk:\\ A Second Cut\thanks{This is the full version of the paper with the same title which will appear in the proceedings of the 26th International Colloquium on Structural Information and Communication Complexity, July 1--4, 2019, L'Aquila, Italy.}\fnmsep\thanks{Yann Disser is supported by the `Excellence Initiative' of the German Federal and State Governments and the Graduate School~CE at TU~Darmstadt.}}
%
%
\author{Yann Disser \and S\"oren Schmitt}
\authorrunning{Y. Disser and S. Schmitt}
%
\institute{Department of Mathematics, TU Darmstadt, Darmstadt, Germany \email{disser@mathematik.tu-darmstadt.de, soeren.schmitt@stud.tu-darmstadt.de}}
\maketitle              
\begin{abstract}
We present an improved algorithm for the problem of evacuating two robots from the unit disk via an unknown exit on the boundary. Robots start at the center of the disk, move at unit speed, and can only communicate locally.	Our algorithm improves previous results by Brandt et al.~[CIAC'17] by introducing a second detour through the interior of the disk.	This allows for an improved evacuation time of $5.6234$. The best known lower bound of~$5.255$ was shown by Czyzowicz et al.~[CIAC'15].
\end{abstract}
\section{Introduction}

We consider the problem of evacuating two robots from the unit disk via an unknown exit on the boundary.
The robots start at the center of the disk and move at unit speed (with infinite acceleration).
They have unlimited computing resources and we neglect the time taken to perform arbitrary calculations.
However, robots are point-shaped and only perceive the information available at their respective locations.
In particular, they can only exchange information (in no time) while colocated at the same point on the disk.
Both robots have full knowledge of the algorithms executed by either robot, and they share the same coordinate system.
The objective is to minimize the evacuation time, i.e., the time needed until both robots have reached the exit, in the worst case over all possible positions of the exit.
Note that the evacuation time for an algorithm is equal to its competitive ratio, since the shortest path to any potential exit location has length one.

This evacuation problem was first introduced by Czyzowicz et al.~\cite{A1}, who showed that the basic algorithm that moves both robots along the boundary in opposing directions achieves an evacuation time of~$5.74$ and gave a lower bound of~$5.199$.
In a follow-up paper, Czyzowicz et al.~\cite{A2A3} presented two improved algorithms with evacuation times of~$5.644$ and~$5.628$.
Both these algorithms introduce detours through the interior of the disk and may lead to a forced meeting before the exit is found.
Additionally, Czyzowicz et al.~\cite{A2A3} improved the lower bound to~$5.255$.
Brandt et al.~\cite{A4} introduced a general necessary condition for worst-case exit positions and gave a slightly improved algorithm, without forced meeting, that achieves an evacuation time of~$5.625$.
Figure~\ref{fig:BildAlgs} shows the trajectories of both robots in each of these algorithms.

\begin{figure}[hbt]
\begin{minipage}[b]{0.235\textwidth}
\centering
\scalebox{0.28}{\includegraphics{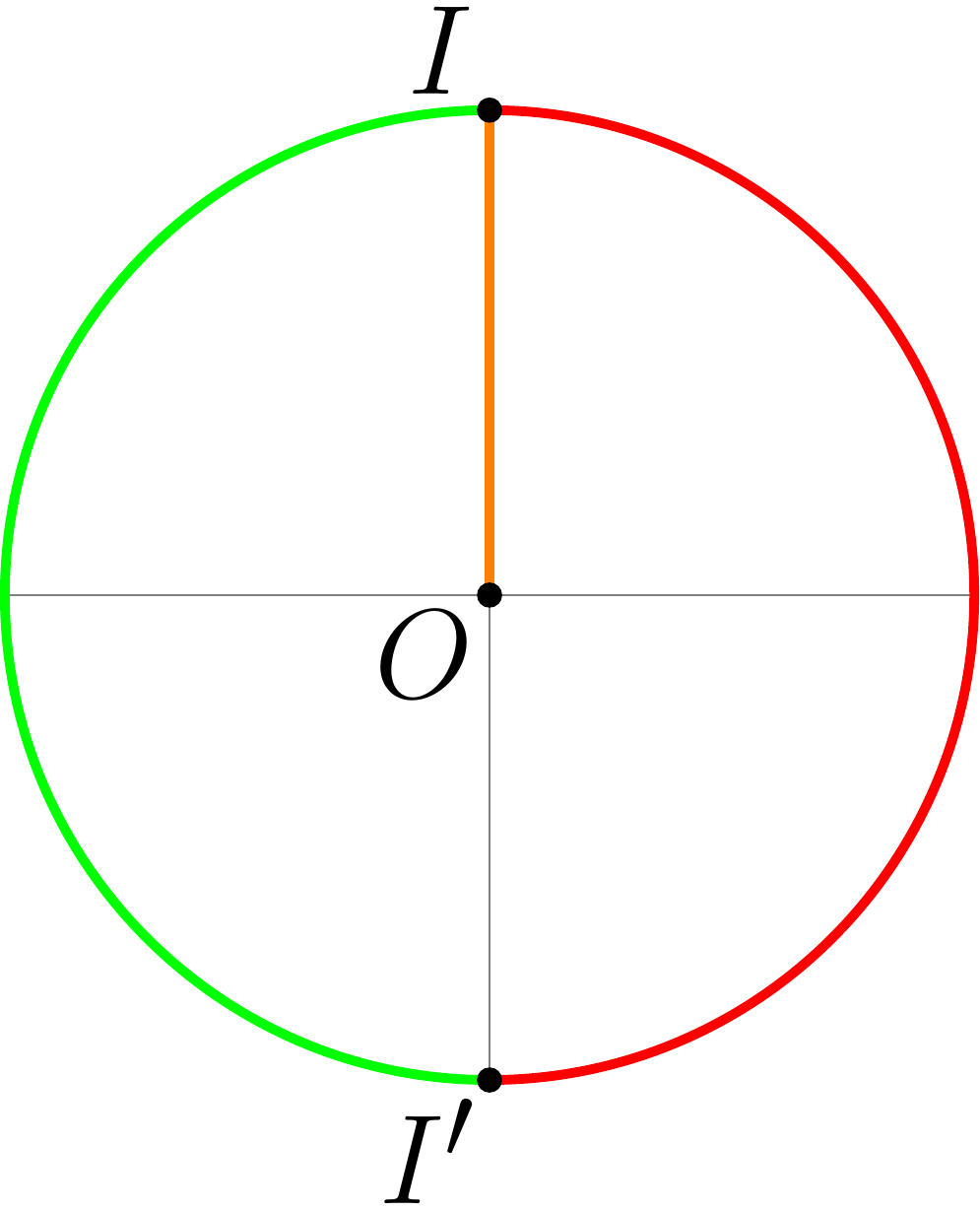}}
\end{minipage}
\hfill
\begin{minipage}[b]{0.235\textwidth}
\centering
\scalebox{0.28}{\includegraphics{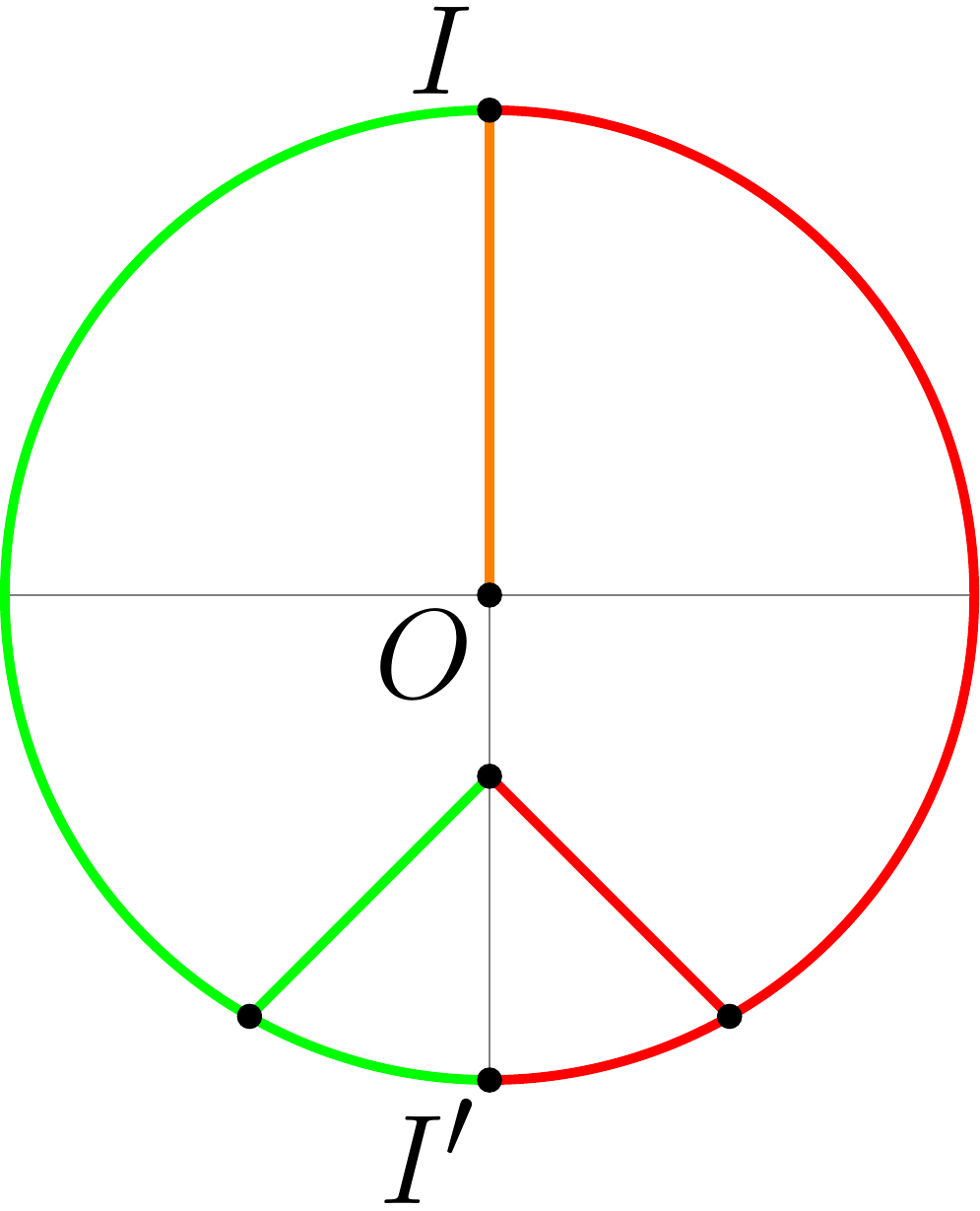}}
\end{minipage}
\hfill
\begin{minipage}[b]{0.235\textwidth}
\centering
\scalebox{0.28}{\includegraphics{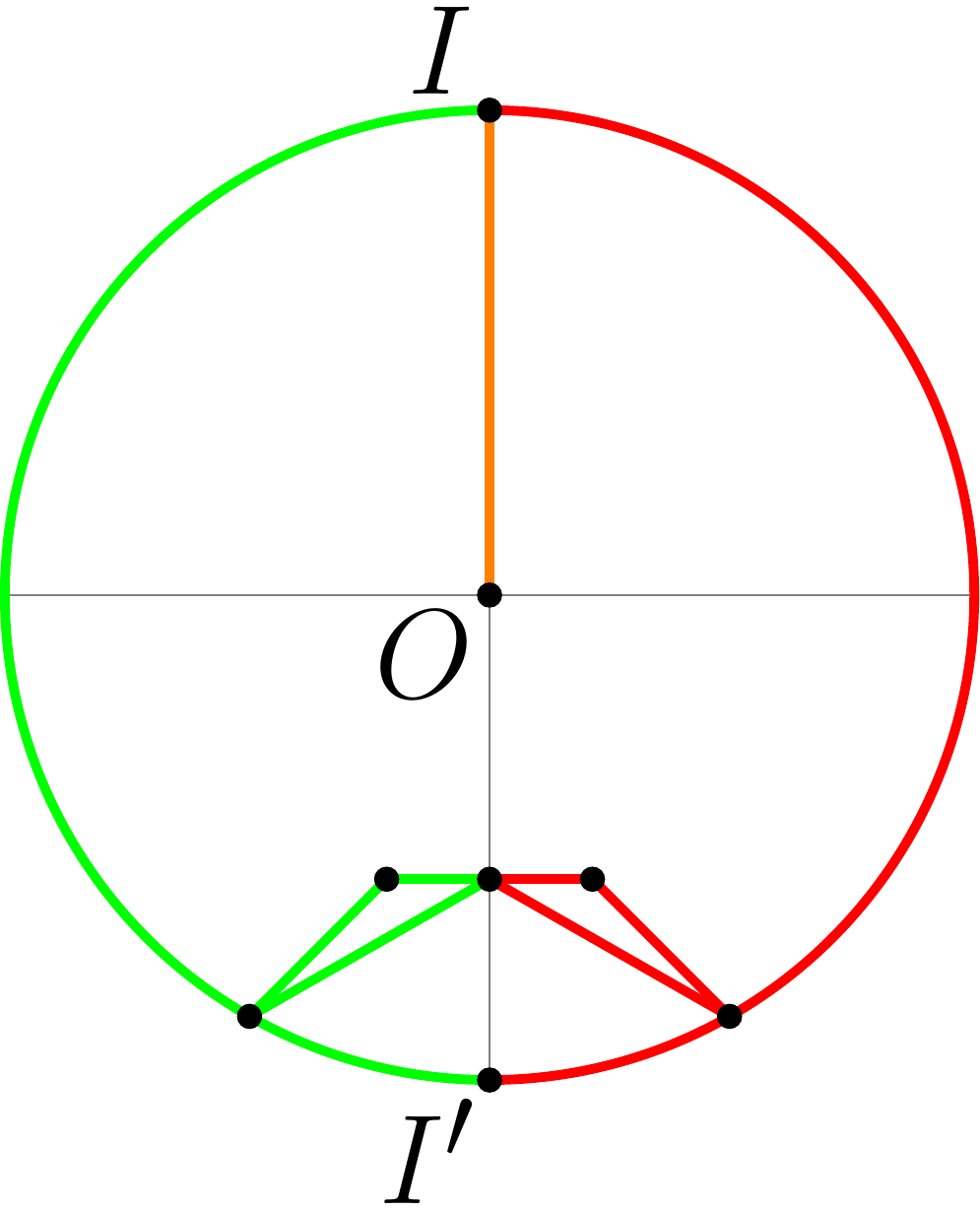}}
\end{minipage}
\hfill
\begin{minipage}[b]{0.235\textwidth}
\centering
\scalebox{0.28}{\includegraphics{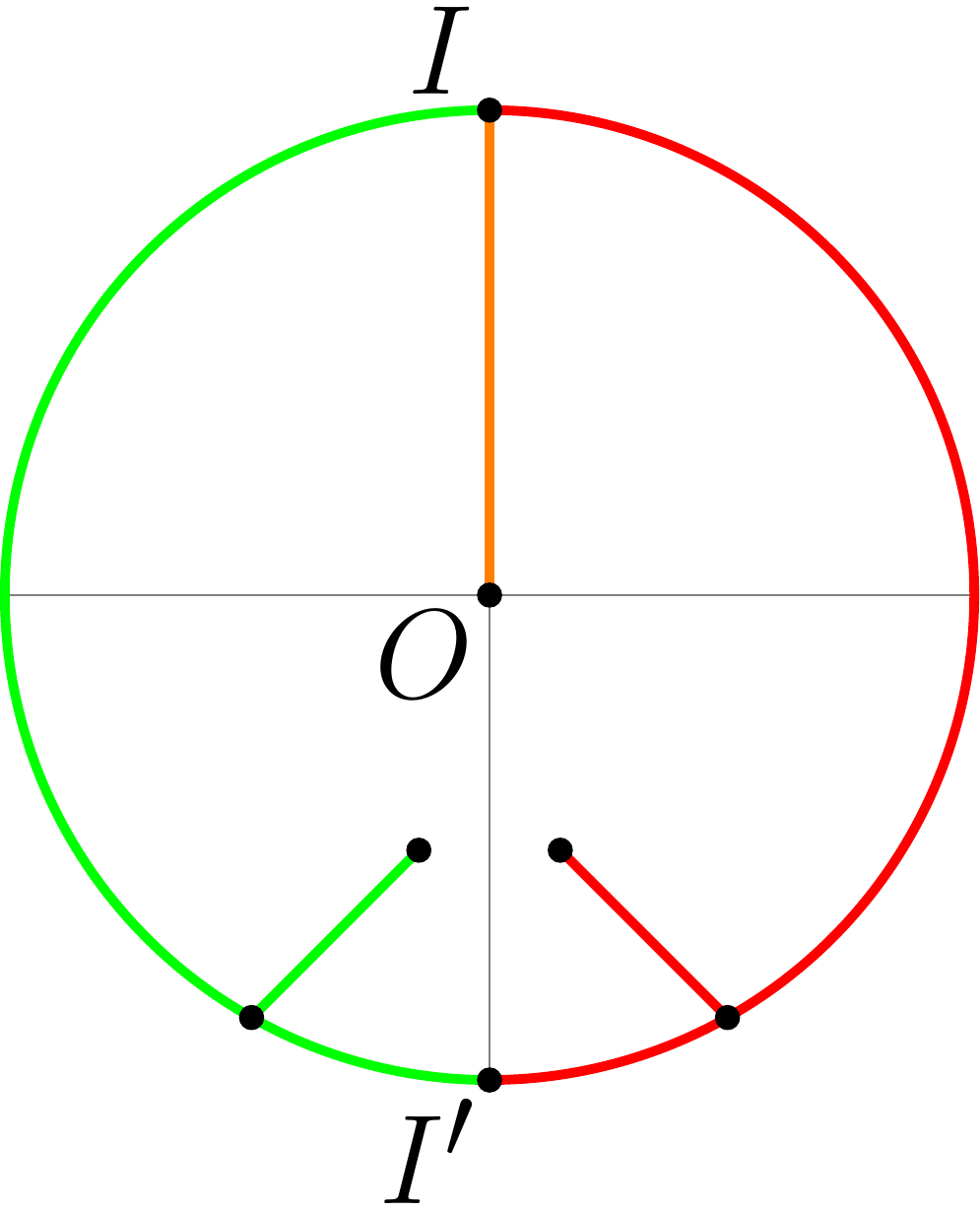}}
\end{minipage}
\caption{From left to right: The evolution of the algorithms presented in~\cite{A1},~\cite{A2A3} and~\cite{A4}. Green indicates the trajectories of robot $R_1$, red indicates the trajectories of robot $R_2$ and orange indicates that both robots move together.}
\label{fig:BildAlgs}
\end{figure}

\subsection{Our Results}

The idea behind introducing a detour through the interior of the disk is to protect the algorithm against the worst-possible exit position: Since the robot that finds the exit needs to intercept the other robot, it makes sense to move towards the other robot before reaching the worst-case interception point, which worsens the evacuation time for some exit positions, but improves it for the worst case. 
Of course, we can apply this idea iteratively to improve the new worst case by introducing a second detour etc.
Brandt et al.~\cite{A4} discuss this idea and state: ``However, the improvement in the evacuation time achieved by the collection of these very small cuts is negligibly small, even compared to the improvement given by our algorithm.''
We refute this statement by showing that introducing even only a single additional cut reduces the evacuation time by the same order of magnitude as the improvement by Brandt et al.~\cite{A4} relative to the result of Czyzowicz et al.~\cite{A2A3}. 
Specifically, we improve the evacuation time to~$5.6234$.
This indicates that there might still be room for improvement in the upper bound when considering a large family of additional cuts. It is worth noting that our algorithm does not use a forced meeting of the agents on either detour to the interior of the disk (see Figure~\ref{fig:Bild1A5}).

\subsection{Related Work}

Robot evacuation has been studied for various settings, differing in number of robots and/or exits, robot capabilities, objective, shape of the region, initial knowledge etc.
Most results were obtained for evacuation from the disk with wireless communication, i.e., for robots that can exchange information at all times.
Czyzowicz et al.~\cite{20} and Pattanayak et al.~\cite{6} consider evacuation with multiple exits and known positions of the exists relative to each other. 
Lamprou et al.~\cite{2} consider two robots of different speeds.
Regarding evacuation with more than two robots, Czyzowicz et al.~\cite{19} study the setting with three robots, one of which may be faulty.
Czyzowicz et al.~focus on evacuating a single robot, the ``queen'', that is supported by up to three~\cite{7} or more~\cite{9} ``servants''.
Regarding other environments, Czyzowicz et al.~\cite{5} study evacuation from equilateral triangles and squares, and Borowiecki et al.~\cite{BorowieckiDasDereniowskiKuszner/16} study the evacuation problem in graphs.

A problem closely related to evacuation is the search problem.
Especially the problem of finding a specific point on the line has received considerable attention, e.g., \cite{BaezayatesCulbersonRawlins/93,14,KaoReifTate/96}, and other works have focused on searching the plane, e.g., \cite{FeinermanKorman/16}.
Another related problem is the rendezvous or gathering problem, where robots initially located at different points need to find each other~\cite{17,ChalopinDasDisserEtal/13b,ChalopinDasDisserEtal/15,KranakisKrizancRajsbaum/06,16}.
Finally, the problem where one robot is trying to catch the other is called the lion and man problem and was first studied for the unit disk~\cite{Littlewood/53}.

\section{Preliminaries}
In this section we define the general notation for the following work. We use the following notation for line segments and arcs between two points~$A$, $B$.
\begin{description}
\item[$\overline{AB}$] denotes the straight line segment between $A$ and $B$.
\item[$\vert \overline{AB}\vert$] denotes the length of the segment~$\overline{AB}$.
\item[$\wideparen{AB}$] denotes the shorter arc from $A$ to $B$ along the boundary of the disk for~$A$ and $B$ on the boundary.
\item[$\vert \wideparen{AB} \vert$] denotes the length of the arc~$\wideparen{AB}$ for~$A$ and $B$ on the boundary.
\end{description}
A \textit{cut} is the movement of a robot from the boundary of the disk into the interior and back to the point where the robot left the boundary. In general a cut can have any shape, but our algorithm only uses line segments.
The \textit{depth} of a cut is defined as half the distance traveled when moving along a cut.
The \textit{evacuation time} is the time until both robots have reached the exit.

Our task is to define trajectories for the robots that minimize the evacuation time. To obtain the evacuation time for a given exit we need to know where the robots exchange the information about the location of the found exit. This is done by the \textit{meeting protocol}, a term coined in \cite{A2A3}. For an illustration, refer to Figure \ref{fig:Bild2A1}.

\begin{definition}[Meeting Protocol]
If at any time $t_0$ one of the robots finds the exit at point $E$, it computes the shortest additional time $t$ so that the other robot, after traveling distance $t_0+t$, is located at point $M$ satisfying $\vert\overline{EM}\vert=t$. This ensures that the robot that found the exit can move along the segment $\overline{EM}$ to pick the other robot up at point $M$ at time $t_0+t$. After both robots meet they evacuate along the segment $\overline{ME}$ via the exit at $E$, resulting in an evacuation time of $t_0+2t$.
\end{definition}

\begin{figure}[hbt]
\centering
\scalebox{0.45}{\includegraphics{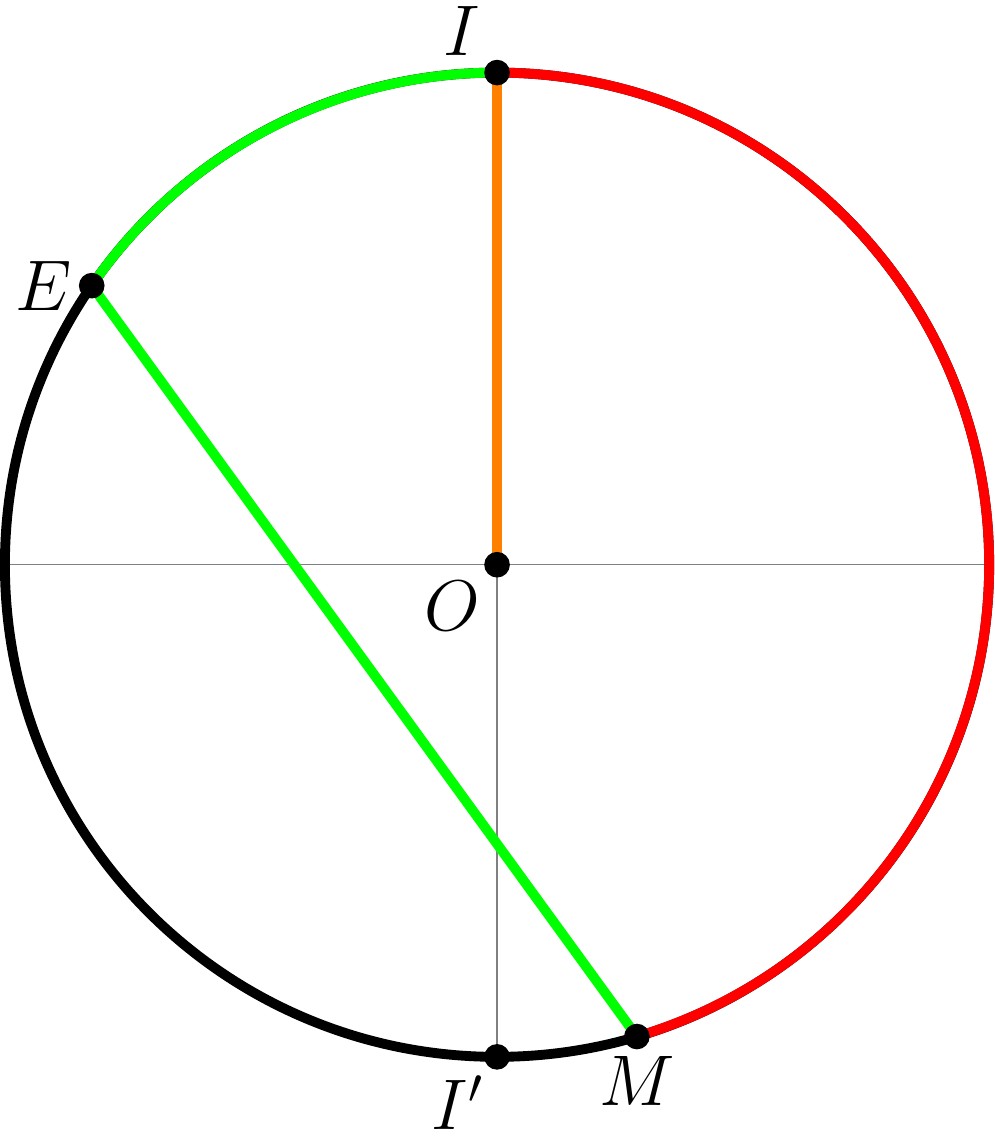}}
\caption{Illustration of the meeting protocol while the robots perform the algorithm presented in~\cite{A1}. Robot $R_1$ finds the exit at point $E$, uses the meeting protocol to calculate the meeting point and picks up robot $R_2$ at point $M$.}
\label{fig:Bild2A1}
\end{figure}

With the meeting protocol we are able to calculate the evacuation time for a given exit. Note that, because the robots move at unit speed, we can use time and traveled distance interchangeably. From the point of view of the robot that finds the exit, the evacuation time is the sum of the time it takes the robot to find the exit, the time it takes the robot to pick up the other robot at their meeting point, and the time it takes to get back to the exit. From the point of view of the robot that gets picked up, the evacuation time is the sum of the time it takes the robot to get to the meeting point and the time it takes the robot to travel from the meeting point to the exit. Note that equations of the form $x+2\sin\left((x+y)/2\right)=y$ have to be solved as part of the meeting protocol. In general there are no closed forms known for these equations and therefore we have to rely on numeric solutions.

To prove our stated evacuation time we will refer to a criterion established in \cite{A4}. We briefly recap the relevant definitions. To properly define certain relevant angles we first distinguish two cases regarding the direction of the movement of the two robots at the exit and the corresponding meeting point.

We say that the movement of the two robots at the exit and the corresponding meeting point is \textit{conform} if the two robots would move to the same side of the infinite line through the exit and the corresponding meeting point if they did not find the exit, respectively were not picked up at the corresponding meeting point. If the two robots would move to different sides of the infinite line we say that their movement (at the exit and the corresponding meeting point) is \textit{converse}. For an illustration, refer to Figure \ref{fig:BildMovement}. The authors of \cite{A4} note that the cases where one or both robots would move \emph{on} the infinite line can be arbitrarily considered to belong to one of the two cases.

For the cases of conform and converse movement we now (under the assumption of local differentiability of the movement) define two angles regarding the movement of the two robots at the exit $E$ and at the corresponding meeting point $M$, respectively, and the straight line segment~$s$ between~$E$ and~$M$.

We assume that, locally, robot $R_1$ arrives at the exit~$E$ via the local linearization of its trajectory~$g$ and would continue on $g$ if it did not find the exit at $E$. Analogously, we assume that, locally, robot $R_2$ arrives at the corresponding meeting point~$M$ via the local linearization of its trajectory~$h$ and would continue on $h$ if it was not picked up at $M$. 

In the case of conform movement, the angle between $g$ and $s$ is denoted by~$\beta$ and the angle between $s$ and $h$ by $\gamma$. In the case of converse movement the angle between $g$ and $s$ is denoted by~$\beta$ and the angle between $h$ and $s$ by $\gamma$. Note that only the definition of $\gamma$ differs in the two cases. For an illustration, refer again to Figure \ref{fig:BildMovement}.
\begin{figure}[hbt]
\begin{minipage}[b]{0.45\textwidth}
\centering
\scalebox{0.75}{\includegraphics{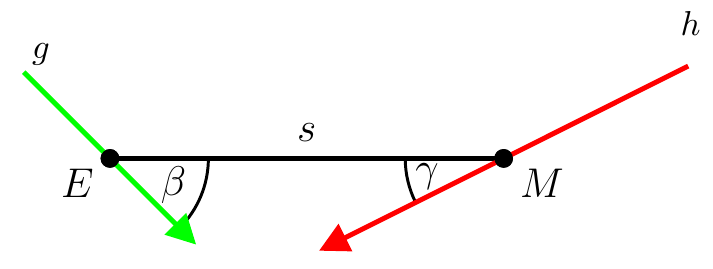}}
\end{minipage}
\hfill
\begin{minipage}[b]{0.45\textwidth}
\centering
\scalebox{0.75}{\includegraphics{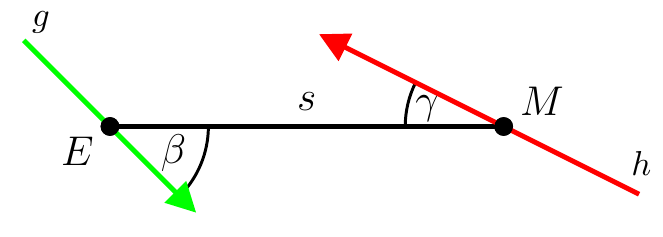}}
\end{minipage}
\caption{Illustration of conform movement (left) and converse movement (right) of the robots.}
\label{fig:BildMovement}
\end{figure}

\begin{theorem}[{\cite[Corollary 2.5, Theorems 2.8, 2.9 and 2.10]{A4}}]\label{notwendigeBed} If the trajectories of the two robots are differentiable around points $E$ and $M$ and $2\cos\left(\beta\right)+\cos\left(\gamma\right)\neq1$ holds, then there is an exit position that yields a larger evacuation time than placing the exit at $E$.
\end{theorem}

This means that, to consider an exit $E$ as the worst-case candidate, it is necessary that either the trajectory of at least one robot is not differentiable (around the exit or the corresponding meeting point) or $2\cos\left(\beta\right)+\cos\left(\gamma\right)=1$.

Like the algorithms illustrated in Figure~\ref{fig:BildAlgs}, our algorithm will follow the idea of initially moving the robots to an arbitrary point $I$ on the boundary (we denote the antipodal point by $I^{\prime}$) and then moving one robot counter-clockwise and the other robot clockwise along the boundary to find the exit. The search on the boundary will only be interrupted if they get picked up or perform a cut. We make a statement about algorithms that follow this general idea, which is helpful to calculate the angles $\beta$ and $\gamma$.
 
\begin{proposition}\label{betagammagleichE} If the robots start their search for the exit together at an arbitrary point $I$ on the boundary, one robot moves counter-clockwise and the other robot moves clockwise, their movement is conform, they move along the boundary towards the point $I^{\prime}$ and their search along the boundary is only interrupted if they get picked up or perform cuts, then the following statement holds: If the corresponding meeting point $M$ of an exit $E$ lies on the boundary, then $\beta=\gamma=\pi-\frac{x+y}{2}$, where $x\coloneqq\vert \wideparen{IE} \vert$ and $y\coloneqq\vert \wideparen{IM} \vert$.
\end{proposition}
\begin{proof} We distinguish between the two cases: $x+y<\pi$ and $x+y\geq \pi$.

For the first case $x+y<\pi$ see the left side of Figure \ref{fig:BildBetaGamma}. Because $\vert\overline{OE}\vert=1$ and $\vert\overline{OM}\vert=1$, the triangle $\triangle EOM$ is isosceles. Therefore the base angles~$\eta$ and~$\eta^{\prime}$ are equal. With the statement, that the interior angles of a triangle add up to~$\pi$, we can express $\eta=\eta^{\prime}$ as $\smash{\frac{\pi-\left(x+y\right)}{2}}$. Our next observation is that the direction vector of the tangent at $E$ and the position vector of $E$ are perpendicular (the same holds for the tangent at $M$ and the position vector of $M$). This immediately yields $\beta=\frac{\pi}{2}+\eta$. Using our expression for $\eta$ we get $\beta=\pi-\frac{x+y}{2}$. The same follows for $\gamma$.

For the second case $x+y\geq\pi$ see the right side of Figure \ref{fig:BildBetaGamma}. We again have an isosceles triangle. Therefore the base angles $\eta$ and $\eta^{\prime}$ are equal. Because of $x+y\geq\pi$ we have to calculate the interior angle of the triangle $\triangle EOM$ at $O$. We easily obtain the interior angle by $2\pi-\left(x+y\right)$. Now analogous to the first case, with the equality of the base angles and the statement about the interior angles of a triangle, we can express $\eta$ as $\left(\pi-\left(2\pi-\left(x+y\right)\right)\right)/2$. In this case the angle $\beta$ can be obtained by~$\frac{\pi}{2}-\eta$. In conjunction with our expression for $\eta$ we have~$\beta=\frac{\pi}{2}-\eta=\frac{\pi}{2}-\left(\pi-\left(2\pi-\left(x+y\right)\right)\right)/2=\pi-\frac{x+y}{2}$. The same follows for~$\gamma$ with the equality of the base angles.
\end{proof}
\begin{figure}[hbt]
\begin{minipage}[b]{0.45\textwidth}
\centering
\scalebox{0.45}{\includegraphics{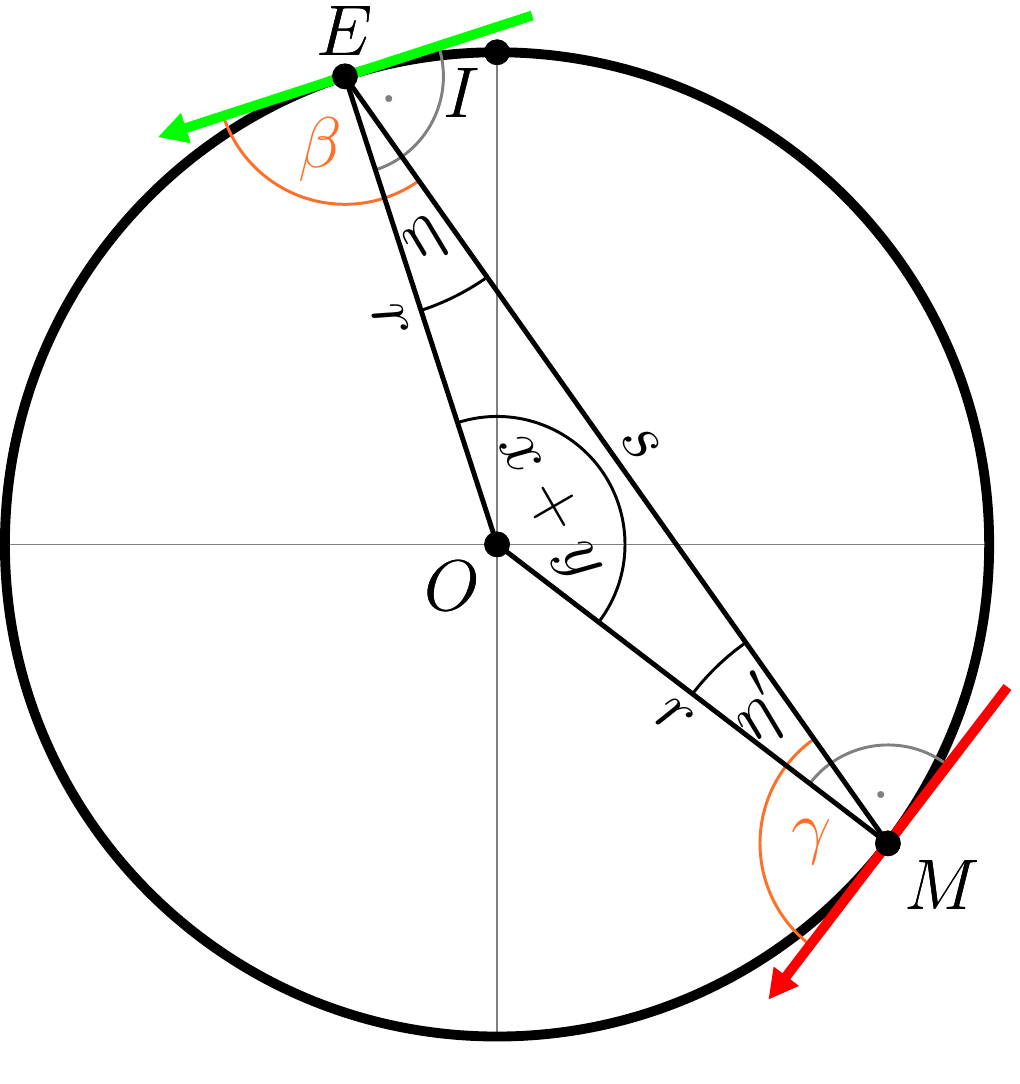}}
\end{minipage}
\hfill
\begin{minipage}[b]{0.45\textwidth}
\centering
\scalebox{0.45}{\includegraphics{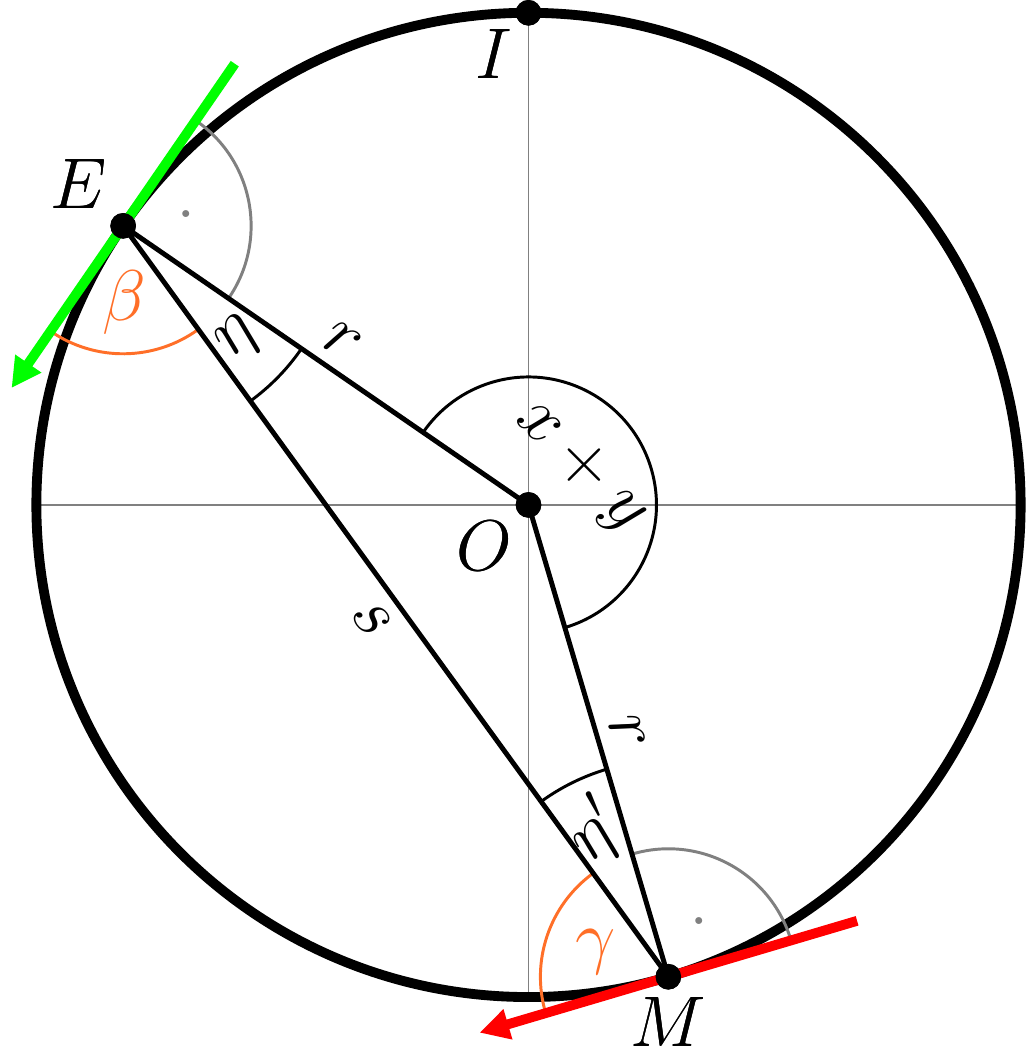}}
\end{minipage}
\caption{Illustration of the cases $x+y<\pi$ (left) and $x+y\geq\pi$ (right) in the proof of Proposition~\ref{betagammagleichE}.}
\label{fig:BildBetaGamma}
\end{figure}

We also make a straight-forward observation (about algorithms that follow the mentioned idea) that follows from monotonicity of the trajectories of both robots along the perimeter of the disk.

\begin{observation}\label{obs1} For any two exit positions $E$ and $E^{\sim}$ with the corresponding meeting points~$M$ and~$M^{\sim}$ lying on the boundary and with~$\vert\wideparen{IE}\vert > \vert\wideparen{IE^{\sim}}\vert$, we have~$\vert\wideparen{IM}\vert > \vert\wideparen{IM^{\sim}}\vert$.
\end{observation}

\section{Our Algorithm}
In this section we present a class $\Alg(p_1, \alpha_1, d_1, p_2, \alpha_2, d_2)$ of parameterized algorithms with two cuts. The position of the first cut is specified by the parameter~$p_1$, the parameter~$\alpha_1$ specifies the angle of the first cut and~$d_1$ describes the depth of the first cut. Correspondingly,~$p_2$ refers to the position,~$\alpha_2$ refers to the angle and $d_2$~refers to the depth of the second cut. The trajectories of the robots are described below. For an illustration, refer to Figure \ref{fig:Bild1A5}.
\algrule
\textbf{Algorithm} $\Alg(p_1, \alpha_1, d_1, p_2, \alpha_2, d_2)$:
\algrule
\noindent{\emph{If a robot finds the exit at any point, it immediately performs the meeting protocol and picks the other robot up at the calculated meeting point $M$. Both robots reach $M$ at the same time and evacuate via the now known exit on a straight line together. Until this happens the trajectories of the robots are as follows:}}
\begin{algorithmic}[1]
\STATE Both robots move on a straight line to an arbitrary point $I$ on the boundary. 
\STATE At $I$ the robots start their search on the boundary in opposite directions: robot $R_1$ moves counter-clockwise and robot $R_2$ moves clockwise.
\STATE After the robots each covered a distance of $p_1$ on the boundary, robot $R_1$ is at $C_1$ and robot $R_2$ is at~$C_1^{\prime}$. There both robots perform their first cut. They move on a straight line at angle $\alpha_1$ towards the interior (see Figure \ref{fig:Bild1A5}). 
\STATE After they each covered a distance of $d_1$ on the straight line, robot $R_1$ is at~$P_1$ and robot $R_2$ is at $P_1^{\prime}$.
\STATE Now both robots return on the straight line to $C_1$ and $C_1^{\prime}$, respectively.
\STATE After reaching $C_1$ and $C_1^{\prime}$, respectively, they continue their search on the boundary, proceeding as in step $2$. 
\STATE After the robots each covered a distance of $p_2$ on the boundary in total, robot $R_1$ is at $C_2$ and robot $R_2$ is at~$C_2^{\prime}$. There both robots perform their second cut. They move on a straight line at angle $\alpha_2$ towards the interior (see Figure \ref{fig:Bild1A5}).
\STATE After they each covered a distance of $d_2$ on the straight line, robot $R_1$ is at~$P_2$ and robot $R_2$ is at $P_2^{\prime}$.
\STATE Now both robots return on the straight line to $C_2$ and $C_2^{\prime}$ respectively.
\STATE After reaching $C_2$ and $C_2^{\prime}$ respectively, they continue their search on the boundary, proceeding as in step $2$.
\end{algorithmic}
\algrule
\begin{figure}[hbt]
\begin{minipage}[b]{0.45\textwidth}
\centering
\scalebox{0.45}{\includegraphics{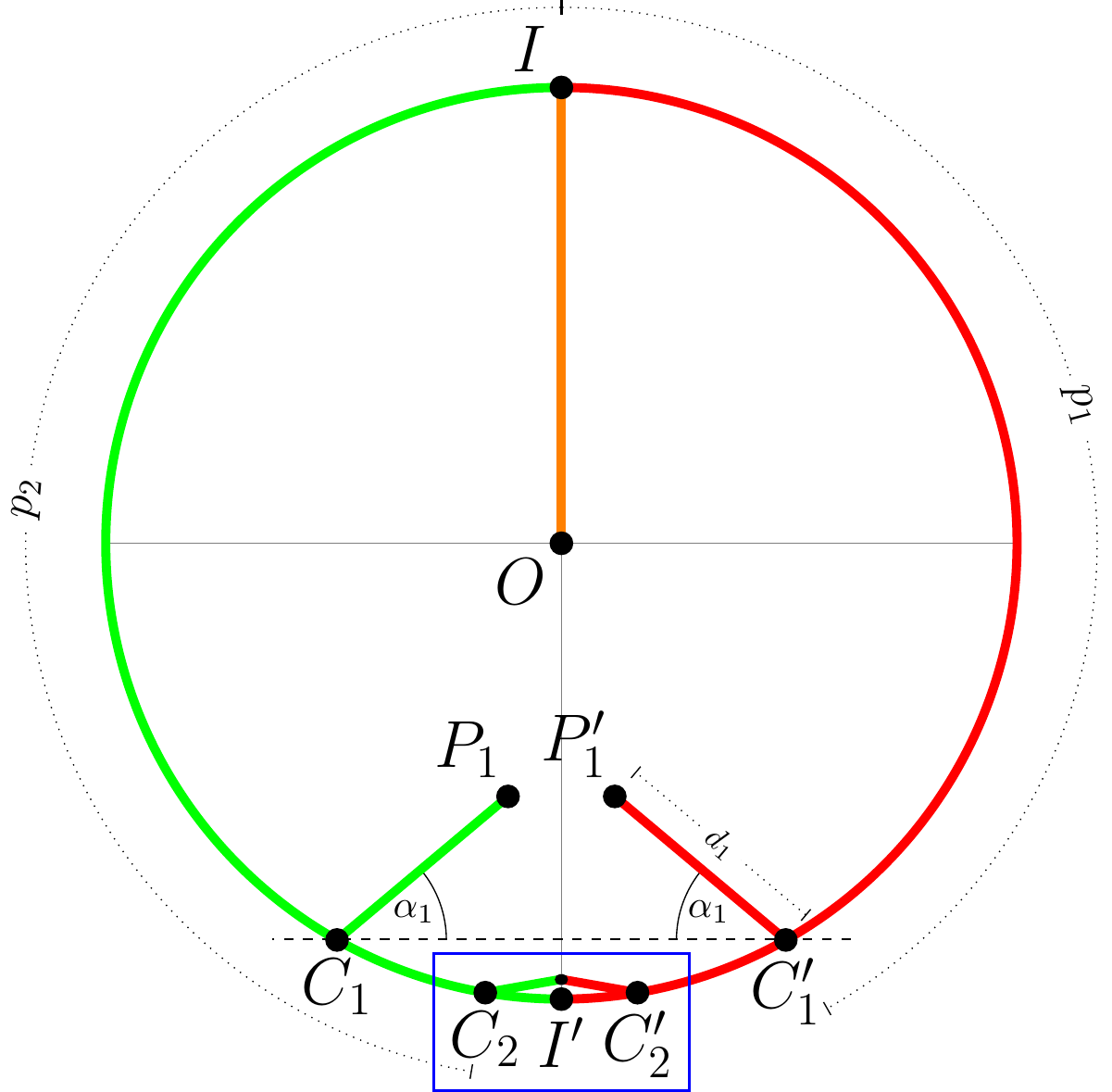}}
\end{minipage}
\hfill
\begin{minipage}[b]{0.45\textwidth}
\centering
\includegraphics[width=\textwidth]{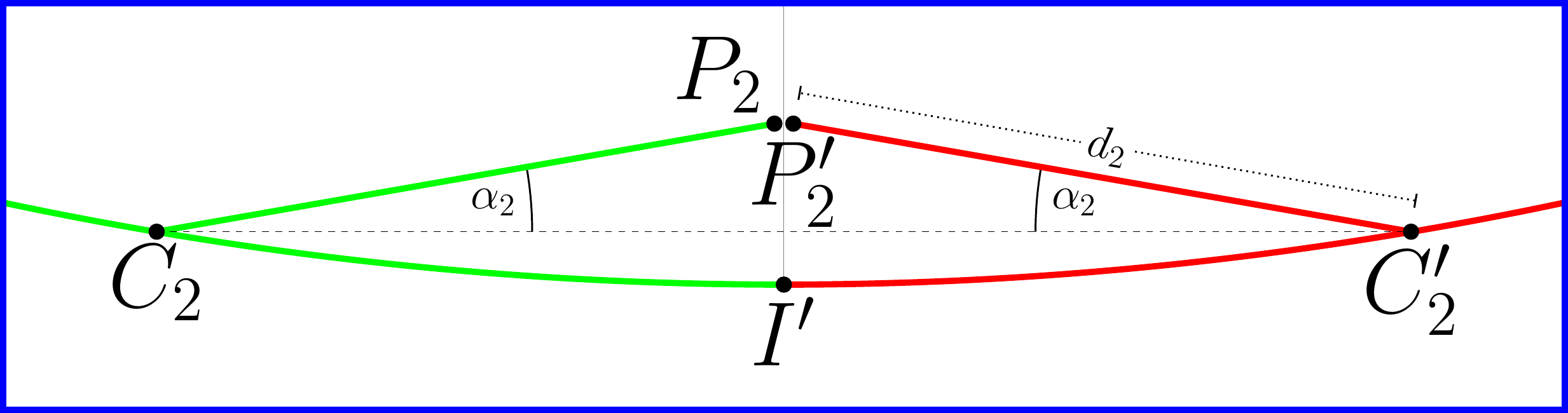}
\end{minipage}
\caption{Illustration of Algorithm $\Alg$. Left: trajectories of the two robots for the exit located at $I^{\prime}$. Right: magnification of the region surrounding~$I^{\prime}$.}
\label{fig:Bild1A5}
\end{figure}

\section{Analysis}
To achieve the stated bound on the evacuation time of $5.6234$ we used local search to computationally determine the parameters $p_1=2.62666582851$, $\alpha_1=2\pi/9$, $d_1=0.490011696287$, $p_2=2.97374843355$, $\alpha_2=0.05523991\pi$ and $d_2=0.1670474016$ for our algorithm. In the following we rigorously prove that our algorithm indeed needs time at most $5.6234$ to evacuate two robots from the unit disk, irrespective of the location of the exit. In order to improve readability, we refer to the above values of our parameters as $p_1^{*}$, $\alpha_1^{*}$, $d_1^{*}$, $p_2^{*}$, $\alpha_2^{*}$ and $d_2^{*}$.

Because the presented algorithm is symmetric, it is sufficient to analyze only one half of the disk. We assume that robot $R_1$ finds the exit and robot $R_2$ gets picked up. For the analysis of our algorithm we proceed along the same arguments as the authors in \cite{A4}, but apply these to our approach with two cuts. First we partition the arc $\wideparen{II^{\prime}}$, into the following parts:
\begin{description}
\item [$\wideparen{IE_1}$] For all exits on this arc robot $R_2$ is picked up before it leaves the boundary at point $C_1^{\prime}$. In particular, if the exit is located at $E_1$, the meeting of the robots will take place at $C_1^{\prime}$.
\item [$\wideparen{E_1E_2}$] For all exits on this arc, robot $R_2$ is picked up while performing its first cut. In particular, if the exit is located at $E_2$, robot $R_2$ will be picked up at~$C_1^{\prime}$ after completing its first cut.
\item [$\wideparen{E_2E_3}$] If the exit lies on this arc, robot $R_1$ finds the exit before performing its first cut, but robot~$R_2$ is picked up after performing its first cut. 
\item [$\wideparen{E_3E_4}$] This part contains all exits for those robot $R_2$ is picked up while performing its second cut. In particular, if the exit is located at $E_4$, robot $R_2$ will be picked up at $C_2^{\prime}$ after completing its second cut. We exclude the point $E_3$ from this (half-open) part.
\item [$\wideparen{E_4E_5}$] If the exit lies on this arc, robot $R_1$ finds the exit before performing its second cut, but robot~$R_2$ is picked up after performing its second cut.
\item [$\wideparen{E_5I^{\prime}}$] This part contains all exits that robot $R_1$ finds after performing its second cut. We exclude the point $E_5$ from this (half-open) part.
\end{description}
We note that all points are well-defined for our parameter values. In particular, if the exit is found after the first cut, the meeting of the robots cannot take place before the other robot performs its second cut. We obtain (for our parameter values): $\vert\wideparen{IE_1}\vert\approx0.629973871925$, $\vert\wideparen{IE_2}\vert\approx2.590020657077$, $\vert\wideparen{IE_3}\vert=2.62666582851$, $\vert\wideparen{IE_4}\vert\approx2.972352082515$ and $\vert\wideparen{IE_5}\vert=2.97374843355$.\\ 
We now analyze each of these parts in detail and determine possible worst-case candidates.
\begin{lemma}\label{A5LemmaIntervall1E} If there is a (global) worst-case exit position on the arc~$\wideparen{IE_1}$, then it is at $E_1$.
\end{lemma}
\begin{proof} To prove the above statement, we will use Theorem \ref{notwendigeBed} to exclude all interior points of the arc~$\wideparen{IE_1}$ as worst-case candidates.

Note that the movement of both robots is not differentiable at $I$ and the movement of robot $R_2$ is not differentiable at $C_1^{\prime}$. Therefore we cannot use Theorem \ref{notwendigeBed} to exclude $I$ and $E_1$ as worst-case candidates, as the corresponding meeting point for an exit at $E_1$ is at $C_1^{\prime}$.

Recall that, for all possible exits on the arc $\wideparen{IE_1}$, robot $R_2$ is picked up before it performs its first cut. In particular, for all these possible exits their corresponding meeting points lie on the boundary. Therefore Proposition \ref{betagammagleichE} is applicable. Also note that we are in the case of conform movement.

With Proposition \ref{betagammagleichE} we obtain $\beta=\gamma\approx 1.51327$ for the exit at $E_1$ and the corresponding meeting point at $C_1^{\prime}$, assuming that robot $R_2$ would continue its search on the boundary at point $C_1^{\prime}$. This yields $2\cos\left(1.51327\right)+\cos\left(1.51327\right){< 1}$. Note that by Proposition \ref{betagammagleichE} and Observation \ref{obs1} the obtained value of approximately $1.51327$ is a lower bound on $\beta=\gamma$ for any exit on the interior of the arc $\wideparen{IE_1}$ and its corresponding meeting point. In conjunction with the monotonicity of the cosine function on the interval $\left[0, \pi\right]$ the statement $2\cos\left(\beta\right)+\cos\left(\gamma\right)< 1$ holds for any exit on the interior of the arc $\wideparen{IE_1}$ and its corresponding meeting point. Therefore we can use Theorem \ref{notwendigeBed} to exclude any exit on the interior of the arc $\wideparen{IE_1}$ as possible worst-case candidate. 

It remains to show that the exit at $I$ is not a worst-case candidate. This is obvious since the evacuation time for this placement would only be $1$ because the robots would find the exit together after initially moving to point $I$.
\end{proof}

Before we start with the next arc, we recall a statement by the authors of \cite{A4}, that also applies to the second cut. We decouple the exit and meeting point for a moment and let~$X$ be an exit for which robot $R_2$ would be picked up while it is returning from the tip of a cut to the boundary. Now instead of the corresponding meeting point, we consider $Y$ as the meeting point for the exit at~$X$, where $Y$ is on the corresponding cut. Let $\beta^{\prime}$ be the angle between the direction of movement of robot~$R_1$ at $X$ and $\overline{XY}$ and let $\gamma^{\prime}$ be the angle between $\overline{XY}$ and the direction of movement of robot~$R_2$ while returning to the boundary along the corresponding cut. They state:
\begin{lemma}\label{claim1} If $2\sin\left(\beta^{\prime}\right)-\sin\left(\gamma^{\prime}\right)>0$, then the value of $2\cos\left(\beta^{\prime}\right)+\cos\left(\gamma^{\prime}\right)$ increases by moving $Y$ by a small $\epsilon$ along the corresponding cut towards the boundary (but not on to the boundary).
\end{lemma}
\renewcommand{\proofname}{Proof sketch}
\begin{proof}
To prove the statement we observe that moving $Y$ in the aforementioned way decreases~$\beta^{\prime}$ by the same amount by which $\gamma^{\prime}$ is increased. Because the cosine function is differentiable, if the $\epsilon$ is small enough the mentioned decrease of $\beta^{\prime}$ gets arbitrarily close to a value that is proportional to the derivative of the function $-\cos\left(\theta\right)$ at $\beta^{\prime}$. Similarly the increase in $\gamma^{\prime}$ gets arbitrarily close to a value that is proportional to the derivative of the function $\cos\left(\theta\right)$ at $\gamma^{\prime}$. Therefore, if $2\sin\left(\beta^{\prime}\right)-\sin\left(\gamma^{\prime}\right)>0$ the value of $2\cos\left(\beta^{\prime}\right)+\cos\left(\gamma^{\prime}\right)$ increases by moving $Y$ in the aforementioned way.
\end{proof}
\renewcommand{\proofname}{Proof}

We now can start with the analysis of the second arc.
\begin{lemma}\label{A5LemmaIntervall2E} If there is a (global) worst-case exit position on the arc $\wideparen{E_1E_2}$, then it is at $E_1$ or $E_2$.
\end{lemma}
\begin{proof} To prove the statement, we further divide the arc into three parts. For this, we calculate the points $Q_1$ and $S_1$, where $Q_1$ is the second intersection of the boundary and the line through $C_1^{\prime}$ and $P_1^{\prime}$. If the exit is placed at $S_1$, robot $R_2$ will be picked up at $P_1^{\prime}$. We obtain: $E_1\approx \left(-0.58912; 0.80804\right)$, $E_2\approx \left(-0.52403;-0.85170\right)$, $Q_1\approx \left(-0.94262; 0.33386\right)$, $S_1\approx \left(-0.82098; -0.57096\right)$, $C_1^{\prime}\approx \left(0.49247; -0.87033\right)$ and $P_1^{\prime}\approx \left(0.11710; -0.55536\right)$.

Now we analyze the three arcs $\wideparen{E_1Q_1}$, $\wideparen{Q_1S_1}$ and $\wideparen{S_1E_2}$ separately.

\begin{description}
\item [$\wideparen{E_1Q_1}$] In this case robot $R_2$ is picked up on its way from $C_1^{\prime}$ to $P_1^{\prime}$. This follows from the fact that $Q_1$ has a greater second coordinate than $S_1$. Also note that the movement of the robots is conform because of the definition of $Q_1$. We want to show, that the statement $2\cos\left(\beta\right)+\cos\left(\gamma\right) > 1$ holds for each exit on the interior of $\wideparen{E_1Q_1}$. For this we observe that the angle $\beta$ of all inner points is at most as large as the corresponding angle for the exit at point $E_1$. In order to prove this, let $E$ be an exit on $\wideparen{E_1Q_1}$ and $M$ be the corresponding meeting point on the line segment $\overline{C_1^{\prime}P_1^{\prime}}$. If the line segment $\overline{E_1C_1^{\prime}}$ is rotated around the origin so far that the point $E$ is reached, denote the point on which $C_1^{\prime}$ is rotated by $Y$. It is easy to see that the point $Y$ lies on the the arc $\wideparen{IC_1^{\prime}}$. Observe in particular that the line segment~$\overline{C_1^{\prime}P_1^{\prime}}$ is not intersected during the rotation due to the definition of $Q_1$. We further observe, that the angle between the tangent at point $E$ and the line segment $\overline{EY}$ is equal to the angle between the tangent at point $E_1$ and the line segment $\overline{E_1C_1^{\prime}}$ (this is the angle with which we will limit the other angles $\beta$ from above). Denote the angle between the tangent at $E$ and the line segment $\overline{EY}$ by $\beta^{\prime}$. Now observe that the angle $\beta$ for an exit at $E$ is smaller than the angle $\beta^{\prime}$ because the line segment $\overline{EY}$ is above the line segment $\overline{EM}$. Therefore, we can limit the respective $\beta$ for any possible exit on $\wideparen{E_1Q_1}$ by the~$\beta$ for the exit at $E_1$, which is approximately $1.51328$ and hence $2\cos\left(\beta\right)>0.11496$. 

On the other hand, the angle $\gamma$ at the corresponding meeting point for each possible exit on interior of the arc $\wideparen{E_1Q_1}$ (and $Q_1$) can be limited from above by the angle~$\angle E_1P_1^{\prime}Q_1 \approx 0.39474$.
Therefore for each possible exit on the interior of the arc $\wideparen{E_1Q_1}$ (and $Q_1$) and the corresponding meeting point the statement $2\cos\left(\beta\right)+\cos\left(\gamma\right) > 2\cos\left(1.51328\right)+\cos\left(0.39474\right)> 0.11496 + 0.92309>1$ holds. With Theorem \ref{notwendigeBed} we can exclude all exits but $E_1$ (because the movement of robot~$R_2$ is not differentiable at the corresponding meeting point $C_1^{\prime}$) on the arc $\wideparen{E_1Q_1}$ as worst-case candidates.

\item[$\wideparen{Q_1S_1}$] We first note that in this case the movement of the robots is converse (because of the definition of~$Q_1$) and the movement of robot $R_2$ is not differentiable around $P_1^{\prime}$ and therefore we cannot exclude the exit at $S_1$ as  the worst-case candidate with Theorem \ref{notwendigeBed}. With a similar reasoning as in the first case we can limit respective $\beta$ of all other points on the arc $\wideparen{Q_1S_1}$ from above by the~$\beta$ for an exit placed at $Q_1$ and its corresponding meeting point. This statement holds because if we rotate the line segment $\overline{Q_1C_1^{\prime}}$ around the origin until we reach a point $E$ on the arc~$\wideparen{Q_1S_1}$ and denote the point on which $C_1^{\prime}$ is rotated on by $Y$, we observe the following: the line segment $\overline{EY}$ is above the line segment $\overline{EM}$, where $M$ is the corresponding meeting point of an exit at $E$. The one thing left to argue is that while we rotated the line segment~$\overline{Q_1C_1^{\prime}}$ we did not intersect the line segment $\overline{C_1^{\prime}P_1^{\prime}}$. This is because the intersection of the line segment $\overline{Q_1C_1^{\prime}}$ and the perpendicular line that contains the origin $O$ is above the point $P_1^{\prime}$. Therefore we can limit the respective $\beta$ by approximately $1.21306$, which is the obtained value of the angle for an exit placed at $Q_1$.
On the other, hand it is easy to verify that the angle $\gamma$ at the corresponding meeting point for each possible exit on the arc $\wideparen{Q_1S_1}$ (but $S_1$, because the movement of robot~$R_2$ is not differentiable at~$S_1$) can be limited from above by the angle~$\angle Q_1P_1^{\prime}S_1 \approx 0.71477$.

Combining these two observations and with the monotonicity of the cosine function on the interval $\left[0, \pi\right]$ we can state that for any possible exit (but $S_1$) on the arc $\wideparen{Q_1S_1}$ the term $2\cos\left(\beta\right)+\cos\left(\gamma\right)$ is greater than $2\cos\left(1.21306\right)+\cos\left(0.71477\right)>0.7003+0.75524>1$. Therefore with Theorem \ref{notwendigeBed} we can exclude all possible exits (but $S_1$) on the arc $\wideparen{Q_1S_1}$ as possible worst-case candidates.

\item[$\wideparen{S_1E_2}$] Here we want to show that $2\cos\left(\beta\right)+\cos\left(\gamma\right)<1$ holds for all possible exits on the interior of the arc $\wideparen{S_1E_2}$. We first observe that the movement of robot $R_2$ is not differentiable at point~$C_1^{\prime}$, which is the corresponding meeting point for $E_2$ and that the movement of the robots is conform. To prove the statement we decouple the exit and the meeting point and use Lemma \ref{claim1}: Let $X$ be an arbitrary point on the arc $\wideparen{S_1E_2}$ and $Y$ be an arbitrary point on the line segment $\overline{P_1^{\prime}C_1^{\prime}}$. First we will show that $2\sin\left(\beta^{\prime}\right)-\sin\left(\gamma^{\prime}\right)>0$ holds for all possible combinations of~$X$ and~$Y$. To do so we verify the combination of $X$ and $Y$ that minimizes $2\sin\left(\beta^{\prime}\right)$ and the combination that maximizes $\sin\left(\gamma^{\prime}\right)$ separately. It is easy to understand that $2\sin\left(\beta^{\prime}\right)$ is minimized for $X=E_2$ and $Y=C_1^{\prime}$ (note that we allow $Y$ to be on the boundary as it is a lower bound for $2\sin\left(\beta^{\prime}\right)$)  and $\sin\left(\gamma^{\prime}\right)$ is maximized for $X=E_2$ and $Y=P_1^{\prime}$. Calculating the respective angles we have: $2\sin\left(\beta^{\prime}\right)-\sin\left(\gamma^{\prime}\right) > 1.01666-0.90489 > 0$. Therefore, we can use Lemma \ref{claim1} and it is enough to show that for every $X$ on the arc $\wideparen{S_1E_2}$ and $Y=C_1^{\prime}$ the statement $2\cos\left(\beta^{\prime}\right)+\cos\left(\gamma^{\prime}\right){<1}$ holds. This is because Lemma \ref{claim1} says that moving $Y$ closer to $C_1^{\prime}$ only increases the value of $2\cos\left(\beta^{\prime}\right)+\cos\left(\gamma^{\prime}\right)$. Therefore, we need to verify the combinations $X$ and $Y=C_1^{\prime}$ that maximize $2\cos\left(\beta^{\prime}\right)$ and $\cos\left(\gamma^{\prime}\right)$. We once again do this separately to obtain an upper bound. It is straight forward to verify that for $X=E_2$ the respective expressions are maximized. We calculate the respective angles (under the assumption that robot $R_2$ would continue his movement at $C_1^{\prime}$ as he did while returning from the tip of its cut) and obtain $2\cos\left(\beta^{\prime}\right)+\cos\left(\gamma^{\prime}\right) < 1.72232-0.77769 < 1$.
Therefore we can use Theorem \ref{notwendigeBed} to exclude all exits on the interior of the arc $\wideparen{S_1E_2}$ as possible worst-case candidates. It remains to show that~$S_1$ is not a worst-case candidate either. We do this by comparing the evacuation times for an exit at $S_1$ and $E_2$. Recall that the corresponding meeting point for an exit at $S_1$ is at~$P_1^{\prime}$. For an exit at~$S_1$ we obtain an evacuation time of approximately $5.05489$ and for an exit at~$E_2$ we obtain an evacuation time of approximately $5.62335779$. Comparing these evacuation times we see that $S_1$ is not a worst-case candidate.
\end{description}
Altogether we have shown that for all possible exits (but $E_1$ and $E_2$) on the arc~$\wideparen{E_1E_2}$ either $2\cos\left(\beta\right)+\cos\left(\gamma\right)>1$ or $2\cos\left(\beta\right)+\cos\left(\gamma\right)<1$ holds. With Theorem~\ref{notwendigeBed} we excluded all possible exits (but $E_1$ and $E_2$) on the arc $\wideparen{E_1E_2}$ as worst-case candidates.
\end{proof}

For the worst-case candidates on the third arc we get an analogous result.
\begin{lemma}\label{A5LemmaIntervall3E} If there is a (global) worst-case exit position on the arc $\wideparen{E_2E_3}$, then it is at $E_2$ or $E_3$.
\end{lemma}
\begin{proof} We first note that the movement of robot $R_2$ is not differentiable at the corresponding meeting point $C_1^{\prime}$  of an exit placed at $E_2$. Furthermore, the movement of robot $R_1$ is not differentiable at~$E_3$. Therefore we cannot use Theorem~\ref{notwendigeBed} to exclude $E_2$ and $E_3$ as worst-case candidates.

Recall that for all possible exits on the arc $\wideparen{E_2E_3}$ robot $R_2$ is picked up after performing its first cut (but before performing its second cut) and robot $R_1$ finds the exit before performing its first cut. In particular, for all these possible exits their corresponding meeting points lie on the boundary. Therefore Proposition~\ref{betagammagleichE} is applicable. Also note that we are in the case of conform movement.

With Proposition \ref{betagammagleichE} we obtain $\beta=\gamma\approx 0.53325$ for the exit at $E_2$ and the corresponding meeting point at $C_1^{\prime}$, assuming that robot $R_2$ continues its search on the boundary at point $C_1^{\prime}$. This yields $2\cos\left(0.53325\right)+\cos\left(0.53325\right)>2.58347> 1$. Note that by Proposition \ref{betagammagleichE} and Observation \ref{obs1} the obtained value of approximately $0.53325$ is an upper bound on $\beta=\gamma$ for any exit on the interior of the arc $\wideparen{E_2E_3}$ and its corresponding meeting point. In conjunction with the monotonicity of the cosine function on the interval $\left[0, \pi\right]$ the statement $2\cos\left(\beta\right)+\cos\left(\gamma\right)> 1$ holds for any exit on the interior of the arc $\wideparen{E_2E_3}$ and its corresponding meeting point. Therefore we can use Theorem~\ref{notwendigeBed} to exclude any exit on the interior of the arc $\wideparen{E_2E_3}$ as a possible worst-case candidate.
\end{proof}
Before we start with the next arc, we recall that we explicitly specified that $E_3$ does not belong to the arc. To still have a closed arc we add an artificial point~$E_3^{\sim}$, that coincides with $E_3$. However, for an exit at $E_3^{\sim}$ robot $R_1$ finds the exit immediately after performing its first cut. Without the addition of the artificial point $E_3^{\sim}$ our arc would be half open and there could be a sequence of points that converges towards $E_3$ with increasing evacuation times but the exit with the largest evacuation time would not belong to the arc. In \cite{A4} the authors observe that the evacuation time for the artificial exit at $E_3^{\sim}$ cannot be smaller than the evacuation time for the exit at $E_3$. This applies, since otherwise robot $R_1$ could improve the evacuation time for the exit at $E_3$ by simulating the movement for the exit at $E_3^{\sim}$.
\begin{lemma}\label{A5LemmaIntervall4E} If there is a (global) worst-case exit position on the arc $\wideparen{E_3^{\sim}E_4}$, then it is at $E_3^{\sim}$ or $E_4$.
\end{lemma}
\begin{proof} To prove the statement, we further divide the arc into three parts. For this, we calculate the points $Q_2$ and $S_2$, where $Q_2$ is the second intersection of the boundary and the line through~$C_2^{\prime}$ and~$P_2^{\prime}$. If the exit is placed at $S_2$, robot $R_2$ will be picked up at $P_2^{\prime}$. We obtain: $E_3^{\sim}=E_3\approx\left(-0.492471164; -0.870328761\right)$, $E_4\approx\left(-0.16843;-0.98571\right)$, $Q_2\approx\left(-0.492471152; -0.870328768\right)$ and also $S_2\approx\left(-0.26963; -0.96296\right)$, $C_2^{\prime}\approx\left(0.16706; -0.98595\right)$ and $P_2^{\prime}\approx\left(0.00252; -0.95710\right)$. Note that~$Q_2$ has a smaller second coordinate than $E_3^{\sim}$ and hence robot $R_1$ reaches the point $Q_2$ after $E_3^{\sim}$.

Now we analyze the three arcs $\wideparen{E_3^{\sim}Q_2}$, $\wideparen{Q_2S_2}$ and $\wideparen{S_2E_4}$ separately.

\begin{description}
\item [$\wideparen{E_3^{\sim}Q_2}$] We first note that our parameters are chosen in such a way that the meeting point $M_3^{\sim}$ for the exit at $E_3^{\sim}$ is on the line segment~$\overline{C_2^{\prime}P_2^{\prime}}$. Once again we limit the respective $\beta$ for any exit on the arc $\wideparen{E_3^{\sim}Q_2}$ from above. To do this we decouple exit and meeting point. Instead of the~$\beta$ for the exit at $E_3^{\sim}$ and the corresponding meeting point at $M_3^{\sim}$, we calculate the angle~$\beta^{\prime}$ for the exit at $E_3^{\sim}$ and consider $C_2^{\prime}$ as the meeting point. We obtain $\beta^{\prime} \approx0.34139$. Note that because of the definition of $Q_2$ the angle~$\beta^{\prime}$ is greater than the respective $\beta$ for the exit at~$E_3^{\sim}$ and the corresponding meeting point at $M_3^{\sim}$. It remains to show that the angle~$\beta^{\prime}$ is an upper bound for the respective angle $\beta$ of any exit on the arc~$\wideparen{E_3^{\sim}Q_2}$. Therefore let $E$ be an exit on the arc $\wideparen{E_3^{\sim}Q_2}$ and $M$ be the corresponding meeting point on the line segment~$\overline{C_2^{\prime}P_2^{\prime}}$. If the line segment $\overline{E_3^{\sim}C_2^{\prime}}$ is rotated around the origin so far that the point $E$ is reached, denote the point on which $C_2^{\prime}$ is rotated by $Y$. Observe in particular that the line segment~$\overline{C_2^{\prime}P_2^{\prime}}$ is not intersected during the rotation due to the definition of $Q_2$. We further observe, that the angle between the tangent at point $E$ and the line segment $\overline{EY}$ is equal to the angle between the tangent at point $E_3^{\sim}$ and the line segment $\overline{E_3^{\sim}C_2^{\prime}}$ (this angle is $\beta^{\prime}$). Because the line segment~$\overline{EY}$ is above the line segment $\overline{EM}$, where $M$ is the corresponding meeting point of the exit at~$E$, the angle $\beta$ for an exit at $E$ is smaller than the angle $\beta^{\prime}$. Therefore, we can limit the respective $\beta$ for any possible exit on $\wideparen{E_3^{\sim}Q_2}$ by $\beta^{\prime}\approx0.34139$.
 
On the other hand, the angle $\gamma$ at the corresponding meeting point for each possible exit of the arc $\wideparen{E_3^{\sim}Q_2}$ can be limited from above by the angle~$\angle E_3^{\sim}P_2^{\prime}Q_2\approx 0.00000001$. 

Therefore with the monotonicity of the cosine function, the following statement $2\cos\left(\beta\right)\allowbreak+\allowbreak\cos\left(\gamma\right)>2\cos\left(0.34139\right)+\cos\left(0.00000001\right)>1.88458+0.99999999>1$ holds for each possible exit on the arc $\wideparen{E_3^{\sim}Q_2}$ and the corresponding meeting point. With Theorem \ref{notwendigeBed} we can exclude all exits but~$E_3^{\sim}$ (because the movement of robot $R_1$ is not differentiable at $E_3^{\sim}$) on the arc~$\wideparen{E_3^{\sim}Q_2}$ as worst-case candidates.

\item [$\wideparen{Q_2S_2}$] We first note that in this case the movement of the robots is converse (because of the definition of~$Q_2$) and the movement of robot $R_2$ is not differentiable around $P_2^{\prime}$ and therefore we cannot exclude the exit at $S_2$ as  the worst-case candidate with Theorem \ref{notwendigeBed}. With a similar reasoning as in the first case we can limit respective $\beta$ of all other points on the arc $\wideparen{Q_2S_2}$ from above by the~$\beta$ for an exit placed at $Q_2$ and its corresponding meeting point. This statement holds because if we rotate the line segment~$\overline{Q_2C_2^{\prime}}$ around the origin until we reach a point~$E$ on the arc~$\wideparen{Q_2S_2}$ and denote the point on which $C_2^{\prime}$ is rotated on by $Y$, we observe the following: the line segment~$\overline{EY}$ is above the line segment~$\overline{EM}$, where $M$ is the corresponding meeting point of an exit at $E$. The one thing left to argue is that while we rotated the line segment~$\overline{Q_2C_2^{\prime}}$ we did not intersect the line segment~$\overline{C_2^{\prime}P_2^{\prime}}$. This is because the intersection of the line segment~$\overline{Q_2C_2^{\prime}}$ and the perpendicular line that contains the origin $O$ is above the point $P_2^{\prime}$. Therefore we can limit the respective $\beta$ by approximately $0.34138552$, which is the obtained value of the angle for an exit placed at $Q_2$.
On the other, hand it is easy to verify that the angle $\gamma$ at the corresponding meeting point for each possible exit on the arc $\wideparen{Q_2S_2}$ (but $S_2$, because the movement of robot~$R_2$ is not differentiable at~$S_2$) can be limited from above by the angle~$\angle Q_2P_2^{\prime}S_2 \approx 0.195072$.

Combining these two observations and with the monotonicity of the cosine function on the interval $\left[0, \pi\right]$ we can state that for any possible exit (but $S_2$) on the arc $\wideparen{Q_2S_2}$ the term $2\cos\left(\beta\right)+\cos\left(\gamma\right)$ is greater than $2\cos\left(0.34138552\right)\allowbreak+\cos\left(0.195072\right)>1.884583+0.98103>1$. Therefore with Theorem \ref{notwendigeBed} we can exclude all possible exits (but $S_2$) on the arc $\wideparen{Q_2S_2}$ as possible worst-case candidates.

\item[$\wideparen{S_2E_4}$] Here we want to show that $2\cos\left(\beta\right)+\cos\left(\gamma\right)<1$ holds for all possible exits on the interior of the arc $\wideparen{S_2E_4}$. We first observe that the movement of robot $R_2$ is not differentiable at point~$C_2^{\prime}$, which is the corresponding meeting point for $E_4$ and that the movement of the robots is conform. To prove the statement we decouple the exit and the meeting point and use Lemma \ref{claim1}: Let~$X$ be an arbitrary point on the arc $\wideparen{S_2E_4}$ and $Y$ be an arbitrary point on the line segment~$\overline{P_2^{\prime}C_2^{\prime}}$. First we will show that $2\sin\left(\beta^{\prime}\right)-\sin\left(\gamma^{\prime}\right)>0$ holds for all possible combinations of~$X$ and~$Y$. To do so we verify the combination of $X$ and $Y$ that minimizes $2\sin\left(\beta^{\prime}\right)$ and the combination that maximizes $\sin\left(\gamma^{\prime}\right)$ separately. It is easy to understand that $2\sin\left(\beta^{\prime}\right)$ is minimized for $X=E_4$ and $Y=C_2^{\prime}$ (note that we allow $Y$ to be on the boundary as it is a lower bound for $2\sin\left(\beta^{\prime}\right)$)  and $\sin\left(\gamma^{\prime}\right)$ is maximized for $X=E_4$ and $Y=P_2^{\prime}$. Calculating the respective angles we have: $2\sin\left(\beta^{\prime}\right)-\sin\left(\gamma^{\prime}\right) > 0.33549 - 0.33289> 0$. Therefore, we can use Lemma \ref{claim1} and it is enough to show that for every $X$ on the arc $\wideparen{S_2E_4}$ and $Y=C_2^{\prime}$ the statement $2\cos\left(\beta^{\prime}\right)+\cos\left(\gamma^{\prime}\right){<1}$ holds. This is because Lemma \ref{claim1} says that moving $Y$ closer to $C_2^{\prime}$ only increases the value of $2\cos\left(\beta^{\prime}\right)+\cos\left(\gamma^{\prime}\right)$. Therefore, we need to verify the combinations $X$ and $Y=C_2^{\prime}$ that maximize $2\cos\left(\beta^{\prime}\right)$ and $\cos\left(\gamma^{\prime}\right)$. We once again do this separately to obtain an upper bound. It is straight forward to verify that for $X=E_4$ the respective expressions are maximized. We calculate the respective angles (under the assumption that robot $R_2$ would continue his movement at $C_2^{\prime}$ as he did while returning from the tip of its cut) and obtain $2\cos\left(\beta^{\prime}\right)+\cos\left(\gamma^{\prime}\right) < 1.971661-0.985099 < 1$.
Therefore we can use Theorem \ref{notwendigeBed} to exclude all exits on the interior of the arc $\wideparen{S_2E_4}$ as possible worst-case candidates. It remains to show that~$S_2$ is not a worst-case candidate either. We do this by comparing the evacuation times for an exit at $S_2$ and $E_4$. Recall that the corresponding meeting point for an exit at~$S_2$ is at~$P_2^{\prime}$. For an exit at~$S_2$ we obtain an evacuation time of approximately $5.39304$ and for an exit at~$E_4$ we obtain an evacuation time of approximately $5.62335779$. Comparing these evacuation times we see that $S_2$ is not a worst-case candidate.
\end{description}
Altogether we have shown that for all possible exits (but $E_3^{\sim}$ and $E_4$) on the arc $\wideparen{E_3^{\sim}E_4}$ either $2\cos\left(\beta\right)+\cos\left(\gamma\right)>1$ or $2\cos\left(\beta\right)+\cos\left(\gamma\right)<1$ holds. With Theorem \ref{notwendigeBed} we excluded all possible exits (but $E_3^{\sim}$ and $E_4$) on the arc $\wideparen{E_3^{\sim}E_4}$ as worst-case candidates.
\end{proof}
We continue with the analysis of the fifth segment.
\begin{lemma}\label{A5LemmaIntervall5E} If there is a (global) worst-case exit position on the arc $\wideparen{E_4E_5}$, then it is at $E_4$ or $E_5$.
\end{lemma}
\begin{proof} We first note that the movement of robot $R_2$ is not differentiable at the corresponding meeting point $C_2^{\prime}$  of an exit placed at $E_4$. Furthermore, the movement of robot $R_1$ is not differentiable at~$E_5$. Therefore we cannot use Theorem~\ref{notwendigeBed} to exclude $E_4$ and $E_5$ as worst-case candidates.

Recall that for all possible exits on the arc $\wideparen{E_4E_5}$ robot $R_2$ is picked up after performing its second cut and robot $R_1$ finds the exit before performing its second cut (but after continuing its search on the boundary after completing its first cut). In particular, for all these possible exits their corresponding meeting points lie on the boundary. Therefore Proposition \ref{betagammagleichE} is applicable. Also note that we are in the case of conform movement.

With Proposition \ref{betagammagleichE} and Observation \ref{obs1} we can obtain an upper bound for the angles $\beta=\gamma$ of any exit on the interior of the arc $\wideparen{E_4E_5}$ and its corresponding meeting point. We do so by calculating the value of $\beta=\gamma$ for the exit placed at $E_4$ and its corresponding meeting point $C_2^{\prime}$ (assuming that robot $R_2$ would continue its search on the boundary at point $C_2^{\prime}$). For this exit and meeting point we obtain $\beta=\gamma\approx 0.16855$. With the monotonicity of the cosine function on the interval $\left[0, \pi\right]$ the statement $2\cos\left(\beta\right)+\cos\left(\gamma\right)>2\cos\left(0.16855\right)+\cos\left(0.16855\right)> 2.95748>1$ holds for any exit on the interior of the arc $\wideparen{E_4E_5}$ and its corresponding meeting point.

Therefore we can use Theorem \ref{notwendigeBed} to exclude any exit on the interior of the arc $\wideparen{E_4E_5}$ as a possible worst-case candidate.
\end{proof}
Analogous to our preparation of Lemma \ref{A5LemmaIntervall4E} we again recall that we explicitly specified that $E_5$ does not belong to the final  arc. To still have a closed arc we add an artificial point $E_5^{\sim}$, that coincides with $E_5$. For an exit at $E_5^{\sim}$ robot $R_1$ finds the exit immediately after performing its second cut. Again the evacuation time for the artificial exit at $E_5^{\sim}$ cannot be smaller than the evacuation time for the exit at $E_5$. This applies, since otherwise robot $R_1$ could improve the evacuation time for the exit at $E_5$ by simulating the movement for the exit at~$E_5^{\sim}$.
\begin{lemma}\label{A5LemmaIntervall6E} If there is a (global) worst-case exit position on the arc $\wideparen{E_5^{\sim}I^{\prime}}$, then it is at $E_5^{\sim}$.
\end{lemma}
\begin{proof} We first note, that the movement of robot $R_1$ is not differentiable for the exit at $E_5^{\sim}$. Therefore we cannot use Theorem \ref{notwendigeBed} to exclude $E_5^{\sim}$ as the worst-case candidate.

Recall that for all possible exits on the arc $\wideparen{E_5^{\sim}I^{\prime}}$ robot $R_2$ is picked up after performing its second cut and robot $R_1$ finds the exit after performing its second cut. In particular, for all these possible exits their corresponding meeting points lie on the boundary. Therefore Proposition \ref{betagammagleichE} is applicable. Also note that we are in the case of conform movement.

With Proposition \ref{betagammagleichE} we obtain $\beta=\gamma\approx 0.08398$ for the exit at $E_5^{\sim}$ and the corresponding meeting point at $M_5^{\sim}$, which satisfies $\vert\wideparen{IM_5^{\sim}}\vert\approx 3.141494005121$ (assuming that robot $R_1$ would continue its search on the boundary at point $E_5^{\sim}$). This yields $2\cos\left(0.08398\right)+\cos\left(0.08398\right)>2.98942> 1$. Note that by Proposition \ref{betagammagleichE} and Observation \ref{obs1} the obtained value of approximately $0.08398$ is an upper bound on $\beta=\gamma$ for any exit on the interior of the arc $\wideparen{E_5^{\sim}I^{\prime}}$ and its corresponding meeting point. In conjunction with the monotonicity of the cosine function on the interval $\left[0, \pi\right]$ the statement $2\cos\left(\beta\right)+\cos\left(\gamma\right)> 1$ holds for any exit on the interior of the arc $\wideparen{E_5^{\sim}I^{\prime}}$ and its corresponding meeting point. Therefore we can use Theorem \ref{notwendigeBed} to exclude any exit on the interior of the arc~$\wideparen{E_5^{\sim}I^{\prime}}$ as a possible worst-case candidate.

It remains to show that the exit at $I^{\prime}$ is not a worst-case candidate. For this placement the exit and its corresponding meeting point would coincide at $I^{\prime}$, since both robots would reach $I^{\prime}$ at the same time. To exclude $I^{\prime}$ as potential worst-case candidate we compare the evacuations time for an exit placed at $I^{\prime}$ and $E_5^{\sim}$. As noted, both robots reach $I^{\prime}$ at the same time, after traveling a distance of~$\pi$ on the boundary each and after performing both cuts. Therefore the evacuation time for the exit placed at $I^{\prime}$ is $1+\pi+2d_1+2d_2\approx 5.45572$.
On the other hand, we obtain an evacuation time for the exit placed at $E_5^{\sim}$ of approximately $5.62335778$. Comparing both evacuation times we can exclude $I^{\prime}$ as the worst-case candidate.
\end{proof}
We summarize the statements of the Lemmas \ref{A5LemmaIntervall1E} through \ref{A5LemmaIntervall6E} as follows.
\begin{theorem}\label{SatzA5E} For Algorithm $\Alg\left(p_1^{*}, \alpha_1^{*}, d_1^{*}, p_2^{*}, \alpha_2^{*}, d_2^{*}\right)$ the worst-case exit position is at $E_1$, $E_2$, $E_3^{\sim}$, $E_4$ or $E_5^{\sim}$.
\end{theorem}

To determine the evacuation time for our algorithm, we simply take the maximum of the evacuation times for these candidates. This leads to our stated upper bound and main result.
\begin{theorem}\label{EZAlg5E} Algorithm $\Alg\left(p_1^{*}, \alpha_1^{*}, d_1^{*}, p_2^{*}, \alpha_2^{*}, d_2^{*}\right)$ needs time at most $5.6234$ to evacuate two robots via an unknown exit on the boundary of the closed unit disk.
\end{theorem}
\begin{proof} For exit $E_1$ we have $\vert\wideparen{IE_1}\vert\approx0.629973871925$ and for the corresponding meeting point $M_1$ we have $\vert\wideparen{IM_1}\vert=2.62666582851$, which results in an evacuation time of less than~$5.62335779$. For the second worst-case candidate $E_2$ we have $\vert\wideparen{IE_2}\vert\approx2.590020657077$ and for the corresponding meeting point $M_2$ we have $\vert\wideparen{IM_2}\vert=2.62666582851$, which results in an evacuation time of less than~$5.62335779$. For exit $E_3^{\sim}$ we have $\vert\wideparen{IE_3^{\sim}}\vert=2.62666582851$ and for the corresponding meeting point $M_3^{\sim}$ we have $\vert\overline{C_2^{\prime}M_3^{\sim}}\vert\approx0.161251676967$, which results in an evacuation time of less than~$5.62335779$. For the fourth worst-case candidate $E_4$ we have $\vert\wideparen{IE_4}\vert\approx2.972352082515$ and for the corresponding meeting point $M_4$ we have $\vert\wideparen{IM_4}\vert=2.97374843355$, which results in an evacuation time of less than~$5.62335779$. For exit $E_5^{\sim}$ we have $\vert\wideparen{IE_5^{\sim}}\vert=2.97374843355$ and for the corresponding meeting point $M_5^{\sim}$ we have $\vert\wideparen{IM_5^{\sim}}\vert=3.141494005121$, which results in an evacuation time of less than~$5.62335778$. The maximum of these evacuation times is limited from above by $5.6234$ and we have proved our stated upper bound on the problem of the two robot evacuation from the closed disk with face-to-face communication.
\end{proof}
%
%
%
\bibliographystyle{splncs04}
\bibliography{Bibliography}
%
\end{document}